\documentclass[aps,pra,twocolumn,a4]{revtex4-1}

\usepackage{amssymb,amsmath,amsthm}

\usepackage{tikz}
\usetikzlibrary{3d,fadings}
\tikzfading [name=radialfade, inner color=transparent!0, outer color=transparent!100]
\usepackage{setspace}
\usepackage{multirow}
\usepackage{xfrac}
\usepackage{bm}
\usepackage{array}
\usepackage{colonequals}

\newtheorem{theorem}{Theorem}

\newtheorem{proposition}{Proposition}
\newtheorem{question}{Question}

\allowdisplaybreaks

\begin{document}

\title{Reality of the quantum state: Towards a stronger $\psi$-ontology theorem}

\author{Shane Mansfield}
\email[]{shane.mansfield@univ-paris-diderot.fr}
\affiliation{Institut de Recherche en Informatique Fondamentale, Universit\'e Paris Diderot - Paris 7}
\affiliation{Department of Computer Science, University of Oxford}


\begin{abstract}
The Pusey-Barrett-Rudolph no-go theorem provides an argument for the reality of the quantum state by ruling out $\psi$-epistemic ontological theories, in which the quantum state is of a statistical nature. It applies under an assumption of preparation independence, the validity of which has been subject to debate. We propose two plausible and less restrictive alternatives: a weaker notion allowing for classical correlations, and an even weaker, physically motivated notion of independence, which merely prohibits the possibility of super-luminal causal influences in the preparation process. The latter is a minimal requirement for enabling a reasonable treatment of subsystems in any theory. It is demonstrated by means of an explicit $\psi$-epistemic ontological model that the argument of PBR becomes invalid under the alternative notions of independence. As an intermediate step, we recover a result which is valid in the presence of classical correlations. Finally, we obtain a theorem which holds under the minimal requirement, approximating the result of PBR. For this, we consider experiments involving randomly sampled preparations and derive bounds on the degree of $\psi$ epistemicity that is consistent with the quantum-mechanical predictions. The approximation is exact in the limit as the sample space of preparations becomes infinite.
\end{abstract}

\maketitle


\section{Introduction}

A number of recent theorems have addressed the issue of the nature of the quantum state \cite{pusey:12,colbeck:12,colbeck:13,hardy:13a,montina:15}. If we suppose that each system has a certain real-world configuration or state of the matter, the description of which we will call its \emph{ontic state}, then we may pose the question of how a system's quantum state relates to its ontic state.

On the one hand, one might consider a pure quantum state as corresponding directly to, or being uniquely determined by, the ontic state, just as the state in classical mechanics (a point in classical phase space) is completely determined by the ontic state or real-world description of the system. On the other hand, the quantum state differs from a classical state of this kind in that in general it may only be used to make probabilistic predictions about the system. Therefore, one might consider that it merely represents our probabilistic, partial knowledge of the ontic state of a system, and moreover that a given ontic state may be compatible with distinct quantum states.

Rather convincing plausibility arguments can be made to support either of these views, which are referred to as \emph{$\psi$-ontic} and \emph{$\psi$-epistemic}, respectively \cite{harrigan:10}. The theorems go a step further and prove that under certain assumptions the $\psi$-epistemic view is untenable. We are especially concerned with the first of these results, the Pusey-Barrett-Rudolph (PBR) theorem \cite{pusey:12}, which has received the most attention and is considered by some to provide the most convincing case for the reality of the quantum state \cite{leifer:14}.

Most of the PBR assumptions are common to the familiar no-go theorems of Bell \cite{bell:64}, Kochen \& Specker \cite{kochen:75}, etc. These we refer to as ontological assumptions, which are briefly summarised in Sec.~\ref{sec:ass} \footnote{We also briefly outline a relationship between $\psi$-ontology and nonlocality/contextuality in the appendix.}. In addition to the common ontological assumptions, each no-go theorem postulates some form of independence for composite systems or observations; for example, for Bell this is the locality assumption, while for PBR it is the novel assumption of preparation independence. The validity of this new assumption has been called into question elsewhere \cite{hall:11,schlosshauer:14,emerson:13,wallden:13,gao:14,miller:14,ducuara:16}. In Sec.~\ref{sec:indass}, we give a precise definition of preparation independence and provide our own critique of its strength, building on earlier work by the author \cite{mansfield:14a,mansfield:13t}. In particular, it is pointed out that the assumption is too strong even to allow for classical correlations in multipartite scenarios.

To address the issue, we propose two plausible and less restrictive alternatives. The first allows for classical correlations mediated through a common past, while the second is a much weaker, physically motivated notion of independence, which, as we will explain, is a minimal requirement for enabling a reasonable treatment of subsystems in any theory. In Sec.~\ref{sec:pbr} we outline the PBR argument and provide a statement of the theorem; but in Sec.~\ref{sec:pbrbreaks} we see that the PBR argument is no longer valid under the weaker notions of independence, and demonstrate this by means of an explicit $\psi$-epistemic toy model. At first, this would appear to re-open the door to the possibility of plausible statistical or $\psi$-epistemic interpretations of the quantum state.

However, in Sec.~\ref{sec:thm} we recover two $\psi$-ontology results which hold under the weaker notions of independence. An intermediate step is to obtain a result which holds in the presence of classical correlations. Our main theorem holds under the minimal notion of independence and approximates the result of PBR. The analysis relies on a proposed experiment involving randomly sampled preparations. The proof makes use of a finite de~Finetti theorem \cite{christandl:09}, establishing a mathematical connection to $\psi$-ontology results. It also supposes that a certain symmetry present at the phenomenological level is reflected at the ontological level. The theorem places bounds on the degree of $\psi$ epistemicity that is consistent with the experimentally testable quantum-mechanical predictions. Our conclusion of approximate $\psi$ ontology improves monotonically as the size of the sample space of preparations increases and is exact in the limit as the sample space becomes infinite.



\section{Ontological Assumptions}
\label{sec:ass}



The first major assumption is that a system has an underlying physical state described by an element $\lambda$ of a measurable space $(\Lambda,\mathcal{L})$, which is referred to as the \emph{ontic state} of the system. This may or may not coincide with the quantum state. The space of ontic states is analogous to classical phase space. Furthermore, it is assumed that each pure quantum state $\psi$ induces a probability distribution $\mu_\psi$ on $(\Lambda,\mathcal{L})$, such that preparation of the system in the quantum state $\psi$ results in an ontic state sampled according to the distribution $\mu_\psi$. It is also assumed that the outcome of any measurement on the system depends solely on its ontic state: the probabilities of obtaining particular outcomes $o \in O$ are determined by \emph{measurement response functions} $\xi_o: \Lambda \rightarrow \mathbb{R}^+$ with the property that, for all $\lambda \in \Lambda$,
\begin{equation}
\sum_{o \in O} \xi_o (\lambda) = 1.
\end{equation}
Note that for the purposes of this article we need only consider measurements with finite outcome sets. Together, these assumptions encode the hypothesis that quantum systems should be described by some (possibly deeper or underlying) ontological theory.
We will be particularly interested in situations in which pure quantum state preparations are drawn from some finite set \footnote{To be technically precise, this set, equipped with its discrete $\sigma$-algebra, itself has the structure of a measurable space $(P,\Sigma)$. It will be convenient to consider probability distributions $\mu$ on $(\Lambda \times P, \mathcal{L} \times \Sigma)$ such that $\mu( L \mid p ) = \mu_p ( L )$ for all $L \in \mathcal{L}$ and $p \in P$. This will allow us to reason in a rigorous manner about ontic state distributions which are conditioned on the preparation setting.}.

The second major assumption is that, whatever this ontological theory might be, its predictions should agree with those of quantum mechanics. No-go theorems arise when it is found that predictions disagree for certain classes of ontological theories. Experimental tests may then be proposed to rule in favour of the predictions of quantum theory on the one hand or certain kinds of ontological theories on the other, though to our current knowledge all experimental evidence points to the correctness of quantum theory. So far, the assumptions posited are, in one form or another, common to the no-go theorems of Bell, Kochen \& Specker, etc.


\subsection{$\psi$-ontic vs $\psi$-epistemic ontologies}
We can already ask the following important question.
\begin{question}\label{q:1}
Is it possible for distinct pure quantum states $\psi,\phi$ to give rise to ontic state distributions $\mu_\psi, \mu_\phi$ which ``overlap'' (see Fig.~\ref{fig:ontdis})?
\end{question}

The question is significant because if the answer is no then we may argue that the quantum state is always uniquely defined by the ontic state, and is in this sense an aspect of physical reality. If, on the other hand, an overlap exists for some pair of distinct, pure quantum states, then there would be instances in which the ontic state or real-world description of the system would not uniquely determine its quantum state. This would provide justification for understanding the quantum state as being of an epistemic nature.

This allows us to distinguish between \emph{$\psi$-ontic} and \emph{$\psi$-epistemic} ontological theories, as those which do or do not admit epistemic overlaps, respectively. The distinction was first formalised in \cite{harrigan:10}, to which the reader is referred for a more detailed discussion.

Precisely, in the present discussion, we will quantify the \emph{epistemic overlap} between pure quantum states $\psi,\phi$ by 
\begin{equation}\label{eq:q}
\omega(\psi,\phi) \colonequals 1-D(\mu_\psi,\mu_\phi) ,
\end{equation}
where $D(\mu_\psi,\mu_\phi)$ is the trace distance between the distributions $\mu_\psi,\mu_\phi$. In instances in which it is clear which quantum states are being considered, we may simply denote this quantity by $\omega$. Suppose that we can prepare a system in one of the quantum states $\psi,\phi$; then $\omega(\psi,\phi)$ may be understood as the minimum probability of obtaining an ontic state which is compatible with both quantum states, and which thus \emph{witnesses} the epistemic nature of the quantum state with respect to the ontological theory in question.

\begin{figure}[htbp]
\caption{\label{fig:ontdis} The existence or non-existence of non-trivial epistemic overlaps for distributions induced by pairs of distinct, pure quantum states characterises (a) \emph{$\psi$-epistemic} and (b) \emph{$\psi$-ontic} ontological theories, respectively. The term epistemic overlap is made precise in Eq.~(\ref{eq:q}).}
\includegraphics[width=\columnwidth]{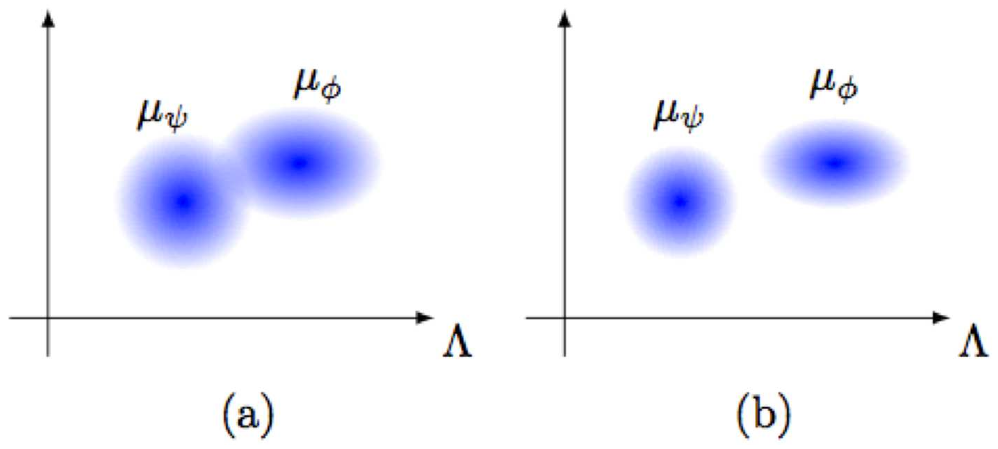}
\end{figure}

\section{Independence assumptions}
\label{sec:indass}

\subsection{Preparation independence}

In addition to the standard ontological assumptions of the previous section, the PBR theorem requires an assumption that systems which are prepared independently should have uncorrelated ontic states.
More precisely, \emph{preparation independence} asserts that the joint probability distribution induced by multipartite preparations is simply the product distribution. It will suffice to consider preparations drawn locally from some finite sets \footnote{These finite sets of pure quantum states or preparations together with their discrete $\sigma$-algebras form measurable spaces $(P_A,\Sigma_A), (P_B,\Sigma_B)$, etc.},
and to describe the property for bipartite systems [as in Fig.~\ref{fig:ind}(a)] consisting of ontic state spaces $(\Lambda_A,\mathcal{L}_A)$ and $(\Lambda_B,\mathcal{L}_B)$. In this setting the assumption can be expressed as
\begin{equation}\label{eq:preind}
\mu(L_A, L_B \mid p_A, p_B) = \mu(L_A \mid p_A) \, \mu(L_B \mid p_B),
\end{equation}
for all $L_A\in \mathcal{L}_A, L_B \in \mathcal{L}_B$ and $p_A,p_B$ pure quantum state preparations in the respective subsystems.
The assumption is perhaps best motivated by the fact that within quantum theory we may consider product quantum states of compound systems, which are completely uncorrelated, so it may therefore be expected that the same situation would hold at the ontological level.

\subsection{Classical correlations}
\label{subsec:compas}

One way in which preparation independence might reasonably be violated, however, is through \emph{classical correlations} arising in the ontic states, mediated through a common past (related points have been raised by Hall \cite{hall:11} and by Schlosshauer and Fine \cite{schlosshauer:14}). Suppose this mediator can take values in an auxiliary measure space $(\Lambda_c,\mathcal{L}_c)$. An effective form of preparation independence might then be recovered by conditioning on this common past [see Fig.~\ref{fig:ind}(b)], resulting in a kind of independence that bears a formal similarity to Bell locality; i.e.,
\begin{equation}\label{eq:compas}
\mu(L_A, L_B \mid p_A, p_B,L_{c}) =
\mu(L_A \mid p_A,L_{c}) \, \mu(L_B \mid p_B,L_{c}),
\end{equation}
for all $L_A\in \mathcal{L}_A, L_B \in \mathcal{L}_B, L_c \in \mathcal{L}_c$, and $p_A,p_B$ pure quantum state preparations in the respective subsystems.

\subsection{Subsystem condition}
\label{subsec:nopresig}

We go a step further and propose an even weaker notion of independence, which we call the \emph{subsystem condition}. The condition is that the ontic state distribution of each subsystem is uncorrelated with the preparation setting of the other(s) [see Fig.~\ref{fig:ind}(c)]; i.e.,
\begin{equation}\label{eq:nopresig}
\begin{aligned}
\mu(L_A \mid p_A, p_B) &= \mu(L_A \mid p_A), \\
\mu(L_B \mid p_A, p_B) &= \mu(L_B \mid p_B),
\end{aligned}
\end{equation}
for all $L_A\in \mathcal{L}_A, L_B \in \mathcal{L}_B$ and $p_A,p_B$ pure quantum state preparations in the respective subsystems.

Otherwise stated, it requires the existence of well-defined marginal probability distributions describing the local behaviour of each preparation device. Without this, it would not even be possible for the theory to describe the preparation of a state on a single subsystem in isolation without having to make reference to all other subsystems---simply put, one could not do local physics.
For instance, Question \ref{q:1} could only be posed for the case of pure states of the entire universe.
We therefore propose that the subsystem condition should be considered a \emph{minimal} notion of independence in preparation scenarios, required by any reasonable ontological theory.

This argument for requiring a minimal notion of independence recalls Einstein's comments on what it would entail to \emph{completely} abolish his ``principle of local action'' (which was conceived in the context of measurement rather than preparation scenarios, but by which we may suppose he had in mind the kind of independence up to classical correlations later formalised by Bell) \cite{einstein:48} \footnote{Translation by the author and Miriam Backens.}:
\begin{quote}
The following idea characterises the relative independence of objects ($A$ and $B$) far apart in space: external influence on $A$ has no \emph{immediate} influence on $B$; this is known as the ``Principle of Local Action'', which is used consistently only in field theory. If this axiom were to be completely abolished,  the existence of (quasi-)isolated systems, and thereby the establishment of laws which can be checked empirically in the accepted sense, would become impossible.
\end{quote}

Furthermore, the subsystem condition can be understood as ruling out, at the ontological level, the possibility of superluminal influences occurring in the preparation process. In this way, the condition is analogous to the assumption of \emph{parameter independence} for ontological models in measurement scenarios, which is borne out at the quantum-mechanical level by the no-signalling theorem \cite{ghirardi:80}. Both preparation independence [Eq.~(\ref{eq:preind})] and the notion of independence up to classical correlations [Eq.~(\ref{eq:compas})] imply the subsystem condition, but not vice versa.

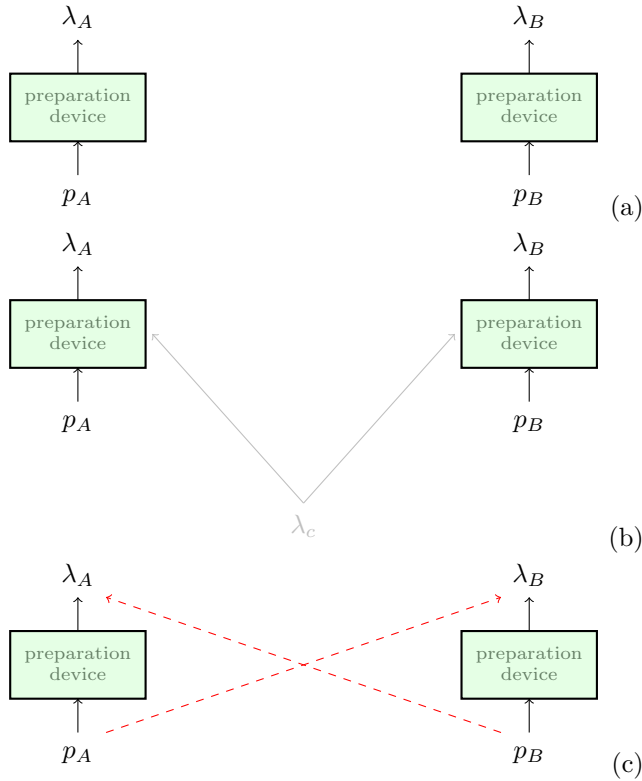
\begin{figure}[htbp]
\caption{\label{fig:ind} (a) Spatially separated, \emph{preparation-independent} devices. (b) A \emph{classical correlation} scenario in which correlations may be mediated through a common past. \label{fig:nopresig} (c) Dashed lines represent correlations forbidden by the \emph{subsystem condition}.}
\begin{center}
\begin{tabular}{cr}
\begin{tikzpicture}[scale=3]
  \path (-.05,0) coordinate (A1);
  \path (-.05,0.3) coordinate (A2);
  \path (0.55,0.3) coordinate (A3);
  \path (0.55,0) coordinate (A4);
  \path (0.25,-0.15) coordinate (B1);
  \path (0.25,0) coordinate (B2);
  \path (0.25,0.3) coordinate (B3);
  \path (0.25,0.45) coordinate (B4);
  \path (0.25,0.15) node {{\scriptsize \begin{tabular}{c} preparation \\ device \end{tabular}}};
  \path (0.25,-0.25) node {$p_A$};
  \path (0.25,0.55) node {$\lambda_A$};
  \begin{scope}[ thick,black]
    \filldraw[fill=green!20,fill opacity=0.5] (A1) -- (A2) -- (A3) -- (A4) -- cycle;
  \end{scope} 
  \draw[->] (B1) to (B2);  
  \draw[->] (B3) to (B4);
    
  \path (1.95,0) coordinate (a1);
  \path (1.95,0.3) coordinate (a2);
  \path (2.55,0.3) coordinate (a3);
  \path (2.55,0) coordinate (a4);
  \path (2.25,-0.15) coordinate (b1);
  \path (2.25,0) coordinate (b2);
  \path (2.25,0.3) coordinate (b3);
  \path (2.25,0.45) coordinate (b4);
  \path (2.25,0.15) node {{\scriptsize \begin{tabular}{c} preparation \\ device \end{tabular}}};
  \path (2.25,-0.25) node {$p_B$};
  \path (2.25,0.55) node {$\lambda_B$};
  \begin{scope}[ thick,black]
    \filldraw[fill=green!20,fill opacity=0.5] (a1) -- (a2) -- (a3) -- (a4) -- cycle;
  \end{scope} 
  \draw[->] (b1) to (b2);  
  \draw[->] (b3) to (b4);
  
    \path (0.375,-.15) coordinate (mans);
  \path (2.125,-.15) coordinate (mbns);
  \path (0.375,.45) coordinate (oans);
\path (2.125,.45) coordinate (obns);
  
\end{tikzpicture}
& (a)
\\
\begin{tikzpicture}[scale=3]
  \path (-.05,0) coordinate (A1);
  \path (-.05,0.3) coordinate (A2);
  \path (0.55,0.3) coordinate (A3);
  \path (0.55,0) coordinate (A4);
  \path (0.25,-0.15) coordinate (B1);
  \path (0.25,0) coordinate (B2);
  \path (0.25,0.3) coordinate (B3);
  \path (0.25,0.45) coordinate (B4);
  \path (0.25,0.15) node {{\scriptsize \begin{tabular}{c} preparation \\ device \end{tabular}}};
  \path (0.25,-0.25) node {$p_A$};
  \path (0.25,0.55) node {$\lambda_A$};
  \begin{scope}[ thick,black]
    \filldraw[fill=green!20,fill opacity=0.5] (A1) -- (A2) -- (A3) -- (A4) -- cycle;
  \end{scope} 
  \draw[->] (B1) to (B2);  
  \draw[->] (B3) to (B4);
    
  \path (1.95,0) coordinate (a1);
  \path (1.95,0.3) coordinate (a2);
  \path (2.55,0.3) coordinate (a3);
  \path (2.55,0) coordinate (a4);
  \path (2.25,-0.15) coordinate (b1);
  \path (2.25,0) coordinate (b2);
  \path (2.25,0.3) coordinate (b3);
  \path (2.25,0.45) coordinate (b4);
  \path (2.25,0.15) node {{\scriptsize \begin{tabular}{c} preparation \\ device \end{tabular}}};
  \path (2.25,-0.25) node {$p_B$};
  \path (2.25,0.55) node {$\lambda_B$};
  \begin{scope}[ thick,black]
    \filldraw[fill=green!20,fill opacity=0.5] (a1) -- (a2) -- (a3) -- (a4) -- cycle;
  \end{scope} 
  \draw[->] (b1) to (b2);  
  \draw[->] (b3) to (b4);
  
    \path (0.375,-.15) coordinate (mans);
  \path (2.125,-.15) coordinate (mbns);
  \path (0.375,.45) coordinate (oans);
\path (2.125,.45) coordinate (obns);

  \path (1.25,-.6) coordinate (p);
  \path (0.58,0.15) coordinate (ma);
  \path (1.92,0.15) coordinate (mb);
  
  \begin{scope}[gray!50]
  \draw[->] (p) to (ma);
  \draw[->] (p) to (mb);
  
  \path (1.25,-.7) node {$\lambda_{c}$};
  \end{scope}
  
\end{tikzpicture}
& (b)
\\
\begin{tikzpicture}[scale=3]
  \path (-.05,0) coordinate (A1);
  \path (-.05,0.3) coordinate (A2);
  \path (0.55,0.3) coordinate (A3);
  \path (0.55,0) coordinate (A4);
  \path (0.25,-0.15) coordinate (B1);
  \path (0.25,0) coordinate (B2);
  \path (0.25,0.3) coordinate (B3);
  \path (0.25,0.45) coordinate (B4);
  \path (0.25,0.15) node {{\scriptsize \begin{tabular}{c} preparation \\ device \end{tabular}}};
  \path (0.25,-0.25) node {$p_A$};
  \path (0.25,0.55) node {$\lambda_A$};
  \begin{scope}[ thick,black]
    \filldraw[fill=green!20,fill opacity=0.5] (A1) -- (A2) -- (A3) -- (A4) -- cycle;
  \end{scope} 
  \draw[->] (B1) to (B2);  
  \draw[->] (B3) to (B4);
    
  \path (1.95,0) coordinate (a1);
  \path (1.95,0.3) coordinate (a2);
  \path (2.55,0.3) coordinate (a3);
  \path (2.55,0) coordinate (a4);
  \path (2.25,-0.15) coordinate (b1);
  \path (2.25,0) coordinate (b2);
  \path (2.25,0.3) coordinate (b3);
  \path (2.25,0.45) coordinate (b4);
  \path (2.25,0.15) node {{\scriptsize \begin{tabular}{c} preparation \\ device \end{tabular}}};
  \path (2.25,-0.25) node {$p_B$};
  \path (2.25,0.55) node {$\lambda_B$};
  \begin{scope}[ thick,black]
    \filldraw[fill=green!20,fill opacity=0.5] (a1) -- (a2) -- (a3) -- (a4) -- cycle;
  \end{scope} 
  \draw[->] (b1) to (b2);  
  \draw[->] (b3) to (b4);
  
    \path (0.375,-.15) coordinate (mans);
  \path (2.125,-.15) coordinate (mbns);
  \path (0.375,.45) coordinate (oans);
\path (2.125,.45) coordinate (obns);

  \draw[->,red,dashed] (mans) to (obns);
  \draw[->,red,dashed] (mbns) to (oans);
  
\end{tikzpicture}
& (c)
\end{tabular}
\end{center}
\end{figure}

\subsection{Other notions of independence and some remarks}

Similarly, in the language of Colbeck and Renner \cite{colbeck:12}, the subsystem condition would be implied by preparation settings being \emph{free} with respect to the causal structure of Fig.~\ref{fig:ind}(a).

Another notion of independence has been proposed by Emerson, Serbin, Sutherland and Veitch \cite{emerson:13}. It is similar to the notion of independence up to classical correlations described in Sec.~\ref{subsec:compas}, but in their case $\lambda_c$ would mediate between the systems without the constraint (\ref{eq:compas}) that justifies considering it as a mediator of classical correlations via a common past. Instead it is required that an effective kind of preparation independence as in Eq.~(\ref{eq:preind}) is obtained if one marginalises to ``forget'' the auxiliary space $\Lambda_c$. Problematically, however, correlations are therefore allowed between $p_A$, $p_B$ and $\lambda_c$ and the subsystem condition is in general violated.
Other works have attempted to drop independence assumptions entirely, with some success, though they obtain correspondingly weaker results that only rule out strong forms of $\psi$ epistemicity \cite{maroney:12,aaronson:13,barrett:14,branciard:14,leifer:14b,knee:16}.

It is implicit in (\ref{eq:preind}) and (\ref{eq:nopresig}) that the set of ontic states in a composite system is simply the Cartesian product of the sets of ontic states in the component subsystems (Leifer refers to this as the Cartesian product assumption \cite{leifer:14}); where they differ is in terms of the constraints placed on global correlations. The situation is somewhat different in classical correlation scenarios; we will return to this issue in Sec.~\ref{sec:thm}. Note also that each of the definitions (\ref{eq:preind})--(\ref{eq:nopresig}) generalises beyond the bipartite scenario in the obvious way.

\section{The PBR Argument}\label{sec:pbr}

The PBR theorem \cite{pusey:12} demonstrates a contradiction between the predictions of quantum theory, on the one hand, and of $\psi$-epistemic, preparation-independent ontological theories on the other. In this section, we briefly outline the argument \footnote{See also \cite{moseley:13,leifer:14} for slightly different versions.}.
To obtain the contradiction, it is supposed that some pair of distinct pure quantum states $\psi,\phi$ give rise to an epistemic overlap \footnote{Since the pair of states will always span a two-dimensional subspace, we may restrict our attention to qubits without loss of generality.}.

\paragraph{$\{\psi,\phi\}$ preparation devices.} Consider a preparation device that is capable of preparing a system in either of these two quantum states. Regardless of which of the preparations is actually made, there is probability at least $\omega >0$ of this device generating an ontic state which witnesses the epistemic overlap, where $\omega$ is defined as in Eq.~(\ref{eq:q}).

\paragraph{Joint systems.} We now consider a second such preparation device, spacelike separated from, and to be operated in parallel with, the first [as in~Fig.~\ref{fig:ind}(a)]. The assumption of preparation independence is invoked to deduce that, regardless of which individual preparations are made, there is probability at least $\omega^2 >0$ of \emph{both} devices generating ontic states which witness the epistemic overlap, and, moreover, that in such a case it should be impossible to make any distinction about which of the joint preparations $$\mathbf{p} \in \{(\psi,\psi),(\psi,\phi),(\phi,\psi),(\psi,\psi)\}$$ has been made.
More generally, we may wish to consider a joint system consisting of $m$ subsystems, each equipped with a $\{\psi,\phi\}$ preparation device, for which this deduction generalises in the obvious way.

\paragraph{Contradiction.} The contradiction arises because quantum theory does, in fact, allow such distinctions to be made. In particular, for all pairs of qubit states, quantum theory admits \emph{conclusive exclusion measurements} \cite{bandyopadhyay:14} on joint systems of the kind described above. In the simplest case, in which only two subsystems are considered, a conclusive exclusion measurement would give outcomes in the set
\begin{equation*}
\{\lnot (\psi,\psi), \lnot (\psi, \phi), \lnot (\phi,\psi), \lnot (\phi, \phi)\},
\end{equation*}
with the property that outcome $\lnot(\psi,\psi)$ has probability zero whenever the joint preparation $(\psi,\psi)$ has been made, so that its occurrence precludes the possibility of the joint preparation having been $(\psi,\psi)$, and so on \footnote{This defining property of conclusive exclusion measurements generalises to $m$-partite systems of this kind in the obvious way. In general, for states with angular separation $\theta$, a quantum-mechanical conclusive exclusion measurement requires at least {$m = {\left\lceil \log_2 \left( \tan \frac{\theta}{2} +1 \right) \right\rceil} ^{-1}$} subsystems \cite{pusey:12,bandyopadhyay:14}.}.

\begin{theorem}\label{thm:pbr}
(PBR \cite{pusey:12}). Suppose there exists a conclusive exclusion measurement for a joint system in which each subsystem is equipped with a $\{\psi,\phi\}$ preparation device. In any preparation-independent ontological theory which can describe this experiment, $\psi$ and $\phi$ have zero epistemic overlap.
\end{theorem}

In other words, epistemic overlaps cannot be reconciled with conclusive exclusion measurements if we accept the assumption of preparation independence. Nevertheless, within quantum theory, such measurements can be specified for all pairs of quantum states. So it follows that any preparation-independent ontological theory which is consistent with quantum theory is necessarily $\psi$ ontic.

\section{When PBR No Longer Applies}\label{sec:pbrbreaks}

In the absence of the assumption of preparation independence, however, it is still possible find $\psi$-epistemic models for the prediction of conclusive exclusion measurements. This is because the reasoning about joint systems from Sec.~\ref{sec:pbr} is no longer valid.

Consider, for example, a pair of $\{ \psi, \phi \}$ preparation devices comprising a joint system which behaves according to Table \ref{tab:couexa}. Restricting our attention to the behaviour of either device in isolation, by considering the corresponding marginal probabilities,
it is clear that there is probability at least $\omega>0$ for the individual device to produce an ontic state witnessing an epistemic overlap of $\psi$ and $\phi$, regardless of which preparation was made. Nevertheless, when the system is considered as a whole, we find that there is zero probability that both ontic states lie in the overlap region at once.

\begin{table}
\caption{\label{tab:couexa} Let $\Delta$ denote the measurable set of ontic states which witness an epistemic overlap of $\psi$ and $\phi$ (the \emph{overlap region}), let $\Delta'$ be its complement, and choose any $\omega_1,\omega_2>0$ such that $\omega_1+\omega_2\leq1$. For each possible pair of preparations $(p_A,p_B)$, the table specifies the probabilities that the resultant ontic states $(\lambda_A,\lambda_B)$ lie within the specified regions.}
\begin{center}
\begin{tabular}{cc||cccc}
$p_A$ & $p_B$ & $(\Delta,\Delta)$ & $(\Delta,\Delta')$ & $(\Delta',\Delta)$ & $(\Delta',\Delta')$ \\ \hline \hline
$\psi$ & $\psi$ & $0$ & $\omega_1$ & $\omega_1$ & $1-2\omega_1$ \\
$\psi$ & $\phi$ & $0$ & $\omega_1$ & $\omega_2$ & $1-\omega_1-\omega_2$ \\
$\phi$ & $\psi$ & $0$ & $\omega_2$ & $\omega_1$ & $1-\omega_1 - \omega_2$ \\
$\phi$ & $\phi$ & $0$ & $\omega_2$ & $\omega_2$ & $1-2\omega_2$ \\
\end{tabular}
\end{center}
\end{table}

The behaviour described in Table \ref{tab:couexa} does not satisfy preparation independence. However, it can easily be checked that it does satisfy the subsystem condition (\ref{eq:nopresig}), and indeed it can even be shown to exhibit independence up to classical correlations (\ref{eq:compas}). Moreover, based on this behaviour, it is not difficult to contrive a $\psi$-epistemic ontological model which reproduces precisely the quantum-mechanical predictions of the conclusive exclusion experiment proposed by PBR, which will thus satisfy both (\ref{eq:compas}) and (\ref{eq:nopresig}).

\subsection{$\psi$-epistemic Model for PBR} If we set $\Lambda := \{ \lambda_1,\lambda_2,\lambda_3 \}$, let $\mu_\psi$ have support $\{ \lambda_1,\lambda_3 \}$, and let $\mu_\phi$ have support $\{ \lambda_2, \lambda_3 \}$ (i.e., $\Delta = \{ \lambda_3 \}$), then Table \ref{tab:couexa} fully determines the behaviour of the preparation devices. We choose the following measurement response functions.
\begin{equation}
\begin{aligned}
\xi_{\lnot (\psi,\psi)}(\bm{\lambda}) &:= \begin{cases} 1 & \mbox{if } \bm{\lambda} = (\lambda_{2},\lambda_{2}) \\ \sfrac{1}{2} & \mbox{if } \bm{\lambda} \in \{(\lambda_{3},\lambda_{2}),(\lambda_{2},\lambda_{3})\} \\ 0 & \mbox{otherwise} \end{cases} \\
\xi_{\lnot (\psi, \phi)}(\bm{\lambda}) &:= \begin{cases} 1 & \mbox{if } \bm{\lambda} = (\lambda_{2},\lambda_{1}) \\ \sfrac{1}{2} & \mbox{if } \bm{\lambda} \in \{(\lambda_{3},\lambda_{1}),(\lambda_{2},\lambda_{3})\} \\ 0 & \mbox{otherwise} \end{cases} \\
\xi_{\lnot (\phi,\psi)}(\bm{\lambda}) &:= \begin{cases} 1 & \mbox{if } \bm{\lambda} = (\lambda_{1},\lambda_{2}) \\ \sfrac{1}{2} & \mbox{if } \bm{\lambda} \in \{(\lambda_{1},\lambda_{3}),(\lambda_{3},\lambda_{2})\} \\ 0 & \mbox{otherwise} \end{cases} \\
\xi_{\lnot (\phi, \phi)}(\bm{\lambda}) &:= \begin{cases} 1 & \mbox{if } \bm{\lambda} = (\lambda_{1},\lambda_{1}) \\ \sfrac{1}{2} & \mbox{if } \bm{\lambda} \in \{(\lambda_{1},\lambda_{3}),(\lambda_{3},\lambda_{1})\} \\ 0 & \mbox{otherwise} \end{cases}
\end{aligned}
\end{equation}
The quantum-mechanical predictions for the PBR experiment \cite{pusey:12} (Table \ref{tab:pbr}) are then reproduced by the operational probabilities
\begin{equation}
p(o \mid p_A, p_B) := \sum_{\lambda_A,\lambda_B \in \Lambda} \, \xi_o (\lambda_A,\lambda_B) \, \mu(\lambda_A,\lambda_B \mid p_A, p_B),
\end{equation}
when
$\omega_1 = \omega_2 = \sfrac{1}{2}$ \footnote{In fact, the \emph{possibilistic} information --- i.e., the supports of the probability distributions --- contained in the empirical model in Table \ref{tab:pbr} is constant for all $\omega_1,\omega_2$. Moreover, this possibilistic information is sufficient for Theorem \ref{thm:pbr} to apply. This is somewhat analogous to the notion of possibilistic or logical nonlocality/contextuality (e.g.~\cite{abramsky:11,mansfield:11}).}.

\begin{table}
\caption{\label{tab:pbr} Probability table (or \emph{empirical model} to use the terminology introduced in \cite{abramsky:11}) representing the quantum-mechanical predictions for the probabilities of observing each of the four outcomes (columns) to the PBR conclusive exclusion measurement, for each possible joint preparation (rows).}
\begin{center}
\begin{tabular}{cc||cccc}
$p_A$ & $p_B$ & $\lnot (\psi,\psi)$ & $\lnot (\psi,\phi)$ & $\lnot (\phi,\psi)$ & $\lnot (\phi,\phi)$ \\ \hline \hline
$\psi$ & $\psi$ & $0$ & $\sfrac{1}{4}$ & $\sfrac{1}{4}$ & $\sfrac{1}{2}$ \\
$\psi$ & $\phi$ & $\sfrac{1}{4}$ & $0$ & $\sfrac{1}{2}$ & $\sfrac{1}{4}$ \\
$\phi$ & $\psi$ & $\sfrac{1}{4}$ & $\sfrac{1}{2}$ & $0$ & $\sfrac{1}{4}$ \\
$\phi$ & $\phi$ & $\sfrac{1}{2}$ & $\sfrac{1}{4}$ & $\sfrac{1}{4}$ & $0$ \\
\end{tabular}
\end{center}
\end{table}

\section{Towards a Stronger Theorem}
\label{sec:thm}

In this section, we establish $\psi$-ontology results which hold under the less restrictive notions of independence introduced in Sec.~\ref{sec:indass}, beginning with the case in which preparations are only assumed to be independent up to classical correlations (\ref{eq:compas}).

First, however, some care is required in order to properly treat classical correlation scenarios [Fig.~\ref{fig:ind}(b)]. What previously constituted a full ontological description $(\lambda_A,\lambda_B) \in \Lambda \times \Lambda$ of the combined system has, in a classical correlation scenario, been supplemented with an additional element $\lambda_c \in \Lambda_c$, whose status it is necessary to address before proceeding. There are two consistent approaches available to us.

\begin{enumerate}
\item\label{pastontic}
We may suppose that $\lambda_c$ is an integral part of the ontic description, and accordingly that the pair $(\lambda_i, \lambda_c)$ --- rather than just $\lambda_i$ --- is statistically sufficient for determining the outcome probabilities for any measurement on the subsystem labelled by $i$.
\item\label{pastepi}
Alternatively, we may suppose that $\lambda_c$ merely serves as a classical mediator between the $\lambda_i$s, and that $\lambda_i$ is still statistically sufficient for determining the outcome probabilities for any measurement on the subsystem labelled by $i$.
\end{enumerate}

Both of these approaches are deserving of consideration; however, for the remainder of this article, we will commit to the \emph{integral-past} approach \ref{pastontic}, and this is the context in which the results in this Sec.~will hold. In this case, the issue of nature of the quantum state no longer comes down the (non-)existence of overlaps for the distributions $\mu_\psi, \mu_\phi$ on $(\Lambda, \mathcal{L})$, but rather for distributions $\mu'_\psi, \mu'_\phi$ on $(\Lambda \times \Lambda_c, \mathcal{L} \times \mathcal{L}_c)$.


\begin{proposition}\label{prop:compas}
\footnote{A version of this proposition was proved by Adam Brandenburger in a private correspondence with the author.}. Suppose there exists a conclusive exclusion measurement for a joint system in which each subsystem is equipped with a $\{\psi,\phi\}$ preparation device. For any ontological theory which can describe this experiment and which is preparation-independent up to integral-past classical correlations, the ontic state distributions $\mu_\psi$ and $\mu_\phi$ necessarily have non-overlapping supports.
\end{proposition}

\begin{proof}

Suppose that $\mu'_\psi(L') > 0$ for some $L' = (L,L_c) \in \mathcal{L}'$. By (\ref{eq:compas}), we know that conditioned on $L_c$ we have factorisability of the joint probability distributions. With $L_c$ fixed, the proof of Theorem \ref{thm:pbr} may therefore be adapted to apply to the probability distributions $\mu'( - \mid \psi, L_c)$ and $\mu'( - \mid \phi, L_c)$, which as a result must have non-overlapping supports. Now, if $\mu'_\psi(L') = \mu'( L \mid \psi, L_c)> 0$ then it must hold that $\mu'_\phi(L') = \mu'( L \mid \phi, L_c) = 0$. Thus $\mu'_\psi$ and $\mu'_\phi$ have non-overlapping supports.
\end{proof}

We now turn to proving a $\psi$-ontology theorem under the minimal notion of independence, the subsystem condition \eqref{eq:nopresig}. The analysis refers to an experiment similar to the one previously considered, but with an important modification.
In the simplest case, the PBR experiment involved a conclusive exclusion measurement on $m=2$ subsystems. Recall, however, that in general $m$ depends on the angular separation of $\psi$ and $\phi$. We consider an experiment in which a conclusive exclusion measurement is performed on $m$ subsystems equipped with $\{\psi,\phi\}$ preparation devices, with the crucial caveat that these subsystems are uniformly sampled from a larger ensemble of  $n \gg m$ subsystems. Note that the quantum-mechanical predictions are the same regardless of which $m$ subsystems are taken.

While we will not be assuming either preparation independence or independence up to classical correlations, we will require an additional \emph{symmetry assumption}: that the joint behaviour of the $\{\psi,\phi\}$ preparation devices is invariant under permutations of those devices. The quantum-mechanical description is certainly invariant under permutations, and so the symmetry assumption merely supposes that this is also reflected at the ontological level.
More precisely, the joint behaviour of the preparation devices will be described in an ontological theory by some $n$-partite conditional probability distribution $\sigma$. A permutation $\pi \in S^n$ acts on $\sigma$ as follows:
\begin{equation}
\begin{aligned}
\pi \cdot & \sigma ( L_1, \dots , L_n \mid p_1, \dots , p_n ) = \\ & \sigma (L_{\pi^{-1}(1)}, \dots , L_{\pi^{-1}(n)} \mid p_{\pi^{-1}(1)} , \dots , p_{\pi^{-1}(n)}).
\end{aligned}
\end{equation}
The distribution $\sigma$ is said to be symmetric if it is invariant under all permutations $\pi \in S^n$.

This kind of shift in focus from independence to symmetry (or exchangeability) is standard in Bayesian statistics. As well as appearing to be a rather natural assumption in the present context, note that both preparation independence and independence up to classical correlations imply symmetry, but not vice versa. In other words, symmetry is a strictly weaker assumption than either assumption of independence. A simple illustration of this fact is provided for instance by the well-known example of P\'olya's urn \cite{polya:30}.

In order to prove our main result, we will also make use of a generalisation (to conditional probability distributions) of Diaconis and Freedman's finite de~Finetti theorem \cite{diaconis:80}, which is due to Christandl and Toner \cite{christandl:09}. In the present setting, this theorem establishes that any symmetric behaviour of the preparation devices which satisfies the subsystem condition may be \emph{approximated} by a description which is independent up to classical correlations \footnote{In the setting originally envisioned by Christandl \& Toner, that of outcome statistics in measurement scenarios, the subsystem condition (\ref{eq:nopresig}) would correspond to the \emph{no-signalling} condition, which similarly requires the existence of well-defined marginal probabilities.}.
We let $\| \mu - \nu \|$ denote the trace distance between any two (conditional) probability distributions $\mu, \nu$ on the same measurable space \cite{christandl:09}, and use this to quantify the approximation.

\begin{theorem}\label{thm:chr}
(Christandl and Toner \cite{christandl:09}). Suppose that $\sigma$ is a symmetric $n$-partite conditional probability distribution on
\begin{equation}
(\Lambda^{n} \times P^{n}, \mathcal{L}^{n} \times \Sigma^{n})
\end{equation}
which satisfies the subsystem condition, and can thus be marginalised to a unique $m$-partite conditional probability distribution $\mu$
on
\begin{equation}
(\Lambda^{m} \times P^{m}, \mathcal{L}^{m} \times \Sigma^{m}),
\end{equation}
for $m < n$.
There exists a probability distribution $\mu_c$ on a finite set of single-subsystem conditional probability distributions labelled $\Lambda_c$, such that
\begin{multline}\label{eq:tradis}
\left\| \mu - \sum_{\omega \in \Lambda_c} \mu_c(\omega) \mu_\omega^{m} \right\| \leq \\ \min \left( \frac{2m |P| |\Lambda|^{|P|}}{n} , \frac{|P| m (m-1)}{n} \right).
\end{multline}
\end{theorem}

We are now in a position to state and prove a $\psi$-ontology result which holds under the minimal notion of independence.

\begin{theorem}\label{thm:nopresig}
Suppose there exists an $m$-partite $\{\psi,\phi\}$ conclusive exclusion measurement, and suppose, moreover, that the $m$ subsystems on which the conclusive exclusion measurement is to be performed are uniformly sampled from a larger ensemble of $n>m$ such subsystems. For any symmetric ontological theory satisfying the subsystem condition and describing this experiment, the epistemic overlap of $\psi$ and $\phi$ is subject to the bound
\begin{equation}\label{eq:bound}
\omega(\psi,\phi) \leq \min \left( \frac{4 m |\Lambda|^{2}}{n} , \frac{2 m (m-1)}{n} \right).
\end{equation}

\end{theorem}

\begin{proof}
Let $\Delta$ be the measurable set of ontic states witnessing the epistemic overlap of $\mu_\psi,\mu_\phi$. Then
\begin{equation}\label{eq:pr1}
\begin{aligned}
\omega(\psi,\phi) &= \min_{p \in \{ \psi, \phi \}} \mu_p(\Delta) \\
&= \min_{p \in \{ \psi, \phi \}} \pi \cdot \mu(\Delta,\Lambda, \dots, \Lambda \mid p, P, \dots , P),
\end{aligned}
\end{equation}
for all permutations $\pi \in S^m$. The first line follows from a re-statement of Eq.~\eqref{eq:q}, and the second follows from the symmetry assumption together with the definition of marginalisation to a single subsystem.
By Theorem \ref{thm:chr}, we know that there exists a probability distribution $\nu = \sum_{\omega \in \Lambda_c} \mu_c(\omega) \mu_\omega^{m}$ satisfying (\ref{eq:tradis}). Since $\nu$ describes preparation correlations which are independent up to conditioning on $\omega \in \Lambda_c$, Proposition \ref{prop:compas} tells us that
\begin{equation}\label{eq:nuzer}
\nu_p(\Delta) = 0,
\end{equation}
for all $p \in \{ \psi, \phi \}$. Then
\begin{equation}
\begin{aligned}
\omega(\psi,\phi) \leq& \; \mu_p(\Delta) \\
=& \left| \mu_p(\Delta) - \nu_p(\Delta) \right| \\
=& \left| \mu(\Delta,\Lambda, \dots, \Lambda \mid p, P, \dots , P)
\right. \\
&\left.
- \nu(\Delta,\Lambda, \dots, \Lambda \mid p, P, \dots , P) \right| \\
\leq& \left\| \mu - \nu \right\| \\
\leq& \min \left( \frac{4 m |\Lambda|^{2}}{n} , \frac{2 m (m-1)}{n} \right) ,
\end{aligned}
\end{equation}
where the first line again follows from the definition of $\omega$ [Eq.~\eqref{eq:q}], the second from Eq.~(\ref{eq:nuzer}), the third from Eq.~\eqref{eq:pr1}, the fourth from the definition of trace distance, and the final line from Eq.~(\ref{eq:tradis}) and the fact that the cardinality of the set of preparations is $|P| = 2$.
\end{proof}

\section{Conclusion}

All no-go theorems rest upon certain assumptions, the justifications for which should in each case be carefully considered \footnote{From a foundational perspective, even the justifications for the common ontological assumptions might be questioned (e.g.~\cite{fuchs:14}), though in any case the no-go results would remain of importance in placing constraints on how quantum predictions may be simulated.}. In particular, the PBR theorem relies on the assumption of preparation independence, which we have argued here may be too strong. Indeed, it does not allow for any correlations in the global ontic state of a multipartite system, though in light of the theorems of Bell and Kochen and Specker it is known that any ontological theory that accounts for the predictions of quantum mechanics will necessarily have such correlations at the ontological level.

We have seen that it is still possible to obtain $\psi$-ontology results by assuming less restrictive notions of independence, which to varying degrees do allow for correlations in the global ontic state, and which can be more strongly justified than preparation independence. Proposition \ref{prop:compas} allows for independence up to classical correlations, and Theorem \ref{thm:nopresig} applies more generally for any correlations consistent with the the very minimal notion of independence, the subsystem condition \eqref{eq:nopresig}. At the same time it must also be noted that these results inevitably introduce other assumptions (the integral-past and symmetry assumptions), the implications of which remain to be further explored, and whose justifications should equally treated with due circumspection.


From a theoretical perspective, Theorem \ref{thm:nopresig} rules out arbitrarily small $\psi$-epistemic overlaps, since quantum theory can perfectly well describe ensembles of $n$ subsystems equipped with $\{ \psi, \phi \}$ preparation devices for arbitrarily large $n$. Therefore, while it is always possible to contrive $\psi$-epistemic toy models for conclusive exclusion experiments by exploiting correlations in the global ontic state, as demonstrated explicitly for the PBR experiment in Sec.~\ref{sec:pbrbreaks}, any full-blown ontological theory which admits arbitrary composition of (sub)systems and respects our symmetry and integral-past assumptions is necessarily $\psi$-ontic \footnote{The key role of compositionality here is reminiscent of approaches to quantum theory in which composition is treated as a primitive \cite{abramsky:09categorical,coecke:10b}.}. 

From an experimental perspective, of course, it is still meaningful and of importance to test whether it is in fact the quantum-mechanical predictions or those consistent with, in this case, certain $\psi$-epistemic ontological theories that are borne out by observation, and Theorem \ref{thm:nopresig} suggests a means of placing experimentally established bounds on epistemic overlaps.


\section*{Acknowledgements}

The author thanks Samson Abramsky, John-Mark Allen, Jon Barrett, Roger Colbeck, Lucien Hardy, Matty Hoban, Dominic Horsman, Matt Leifer, Matt Pusey, Rui Soares Barbosa, an anonymous referee, and especially Adam Brandenburger, for valuable comments and discussions at various stages of this work, as well as audiences at the Besan\c{c}on \'Epiphymaths and Oxford Philosophy of Physics seminars. This research was carried out both at l'Institut de Recherche en Informatique Fondamentale, Universit\'{e} Paris Diderot - Paris 7, and at the Department of Computer Science, Oxford University. Financial support is gratefully acknowledged from the Fondation Sciences Math\'{e}matiques de Paris, postdoctoral research grant eotpFIELD15RPOMT-FSMP1, Contextual Semantics for Quantum Theory, the Networked Quantum Information Technologies (NQIT) hub, and the Oxford Martin School, University of Oxford via the programme for Bio-inspired Quantum Technologies.




\appendix

\section{$\psi$-ontology in relation to nonlocality and contextuality}

A natural question to ask is how the notion of $\psi$ ontology and the no-go results discussed in this article relate to the notions of nonlocality and contextuality as established by the no-go theorems of, e.g., Bell and Kochen and Specker. A first observation is that all of these notions are couched in terms of ontological theories: the theorems all suppose the ontological assumptions of Sec.~\ref{sec:ass} and identify differing operational predictions between certain classes of ontological theories ($\psi$-epistemic, local, noncontextual, or subclasses of these) and quantum theory.

In \cite{mansfield:14a}, we posed the question of whether objects other than the quantum state might give rise to epistemic overlaps on space of ontic states. For instance, if we accept the ontological assumptions then the ontic state is statistically sufficient for determining the outcome probabilities for any given measurement on the system. So for each measurement the theory determines a family of probability distributions on a set $O$ of outcomes conditional on the ontic state $\lambda \in \Lambda$. From this, by Bayesian inversion, we can also obtain a family of probability distributions on the set of ontic states $\Lambda$ conditional on the outcome $o \in O$, provided that the ontological theory is parameter independent or non-contextual \footnote{As discussed in \cite{mansfield:14a}, a necessary condition for performing the inversion is parameter independence/noncontextuality: the requirement that for all measurements and for all $\lambda \in \Lambda$ the probability distribution over outcomes is independent of which measurements may or may not be made jointly or in context with it. Note that noncontextuality is the more general notion, subsuming parameter independence in the case of a multipartite system in which a context is simply a set of measurements with one for each party.}.
Therefore, we can also meaningfully ask whether measurement outcomes can give rise to epistemic overlaps \footnote{See also \cite{allen:16} for a related analysis of superpositions.}.

Outcome-ontic theories in this sense are precisely the \emph{deterministic, non-contextual} ontological theories; determinism being the requirement that for all measurements and all $\lambda \in \Lambda$ the corresponding distribution over outcomes is a $\delta$-function. The failure of an empirical model (i.e.,a probability table representing either empirical data or theoretical predictions for empirical data)
to be realisable by a deterministic, noncontextual ontological model is precisely an instance of \emph{contextuality}. Equivalently \cite{abramsky:11}, contextual empirical models may be characterised as those which fail to have realisation by a factorisable ontological model, and thus contextuality generalises \emph{nonlocality} (a term which applies in the special case of a multipartite measurement scenario).

Any instance of nonlocality or contextuality in a quantum mechanically realisable empirical model, therefore, provides a contradiction between the predictions of quantum mechanics and those of outcome-ontic ontological theories. The theorems of Bell, Kochen and Specker, etc., which identify such empirical models, can thus be re-interpreted as no-go theorems which rule out outcome-ontology. Note the apparent duality here: while $\psi$-ontology theorems rule out certain $\psi$-\emph{epistemic} theories, nonlocality or contextuality theorems rule out certain outcome-\emph{ontic} theories.


\bibliographystyle{unsrt}
\bibliography{refs1}

\end{document}